\documentclass[runningheads,envcountsame,a4paper]{llncs}
\usepackage{amsmath,amssymb}
\usepackage{tikz}
\usetikzlibrary{chains,backgrounds}
\usetikzlibrary{arrows,automata}
\usepackage{url}
\usepackage{float}
\usepackage{array}
\usepackage{graphicx}
\usepackage{subfig}
\usepackage{slashbox}

\def\modd#1 #2{#1\ ({\rm mod}\ #2)}

\def\fmod#1 #2{#1\ ({\rm mod}\ #2)}

\DeclareMathOperator{\pref}{Pref}

\def\cip{\mathring{R}}
\DeclareMathOperator{\scw}{scw}
\DeclareMathOperator{\sw}{sw}

\makeatletter

\newcommand{\Rmnum}[1]{\expandafter\@slowromancap\romannumeral #1@}
\makeatother

\usepackage[dvips]{epsfig}
\usepackage{epsf}

\frontmatter

\title{Sets Represented as the Length-$n$ Factors of a Word}

\author{Shuo Tan \and Jeffrey Shallit}

\institute{School of Computer Science, University of Waterloo,
Waterloo, ON  N2L 3G1, Canada  \\
\email{ \{s22tan,shallit\}@uwaterloo.ca}
}

\begin{document}

\maketitle

\begin{abstract}
In this paper we consider the following problems:  how many different
subsets of $\Sigma^n$ can
occur as set of all length-$n$ factors of a finite word?
If a subset is representable,
how long a word do we need to represent it?
How many such subsets are represented by words of length $t$?
For the first problem,
we give upper and lower bounds of the form $\alpha^{2^n}$ in the
binary case.
For the second problem, we give a weak upper bound and
some experimental data.
For the third problem, we give a closed-form
formula in the case where $n \leq t < 2n$.
Algorithmic variants of these problems have
previously been studied under the name ``shortest common superstring''.
\end{abstract}

\section{Introduction}

Let $w,x,y,z$ be finite words.
If $w = xyz$, we say that $y$ is a {\it factor} of $w$.
De Bruijn proved~\cite{debruijn1946} the existence of a set of binary
words $(B_n)_{n \geq 1}$ with the property that every binary word of length
$n$ appears as a factor of $B_n$ (and, in fact, appears exactly once in
$B_n$). Here we are thinking of $B_n$ interpreted as a circular word.
For example, consider the case where $n = 2$, where we can take $B_2 =
0011$. Interpreted circularly, the factors of length $2$ of $B_2$ are
$00, 01, 11, 10$, and these factors comprise all the binary words of
length $2$.

However, not every subset of $\{ 0, 1\}^n$ can be represented as the
factors of some finite word. For example, the
set $\{ 00, 11 \}$ cannot equal the set of all factors of any
word $w$ --- interpreted in the ordinary sense or circularly ---
because the set of factors of any $w$ containing both letters
must contain either $01$ or $10$.

This raises the natural question, how many different non-empty subsets $S$
of $\{ 0, 1\}^n$ can be represented as the factors of some
word $w$?  (Note that, unlike \cite{moreno}, we do {\it not\/} insist that each
element of $S$ appear exactly once in $w$.)
We give upper and lower bounds for this
quantity for circular words,
both of the form $\alpha^{2^n}$.  Our upper bound has
$\alpha = \root4\of{10} \doteq 1.78$ while our lower bound has $\alpha
= \sqrt{2} \doteq 1.41$.

If the set of length-$n$ factors of a word $w$ (considered circularly)
equals $S$, we say that $w$ {\it witnesses} $S$.  We study the length
of the shortest witness for subsets of $\{0, 1\}^n$, and give an upper
bound.

Restriction on the length of a witness leads us to another interesting
problem. Let $T(t, n)$ denote the number of subsets of $\{0, 1\}^n$
witnessed by some word of length $t \geq n$. Is there any characterization of
$T(t,n)$? We focus on ordinary (non-circular) words for this question
and derive a closed-form formula for $T(t, n)$ in the case where $n \leq t <
2n$.

{\it Algorithmic} versions of related problems have been widely studied
in the literature under the name ``shortest common superstring''.
For example, Gallant, Maier, and Storer \cite{gallant} proved that
the following decision problem is NP-complete:

\medskip

\noindent Instance:  A set $S$ of words and an integer $K$. \\
\noindent Question:  Is there a word $w$ of length $\leq K$ containing each
word in $S$ (and possibly others) as a factor?

\medskip

However, the combinatorial problems that we study in this paper
seem to be new.

\section{Preliminaries}

Let $\Sigma = \{0, 1\}$ denote the alphabet. Let $F_n(w)$ denote the
set of length-$n$ factors of an ordinary (non-circular) word $w$, and
let
$C_n(w)$ denote the set of length-$n$ factors of $w$ where $w$ is
interpreted circularly.
For example, if $w = 001$, then $F_2(w) = \{00, 01\}$, while if $w =
001$ is interpreted circularly, then $C_2(w) = \{00, 01, 10\}$.

We say that a word $w$ {\it witnesses} (resp., {\it circularly
witnesses}) a subset $S$ of $\Sigma^n$ if $F_n(w) = S$ (resp.,
$C_n(w) = S$). A subset $S$ of $\Sigma^n$ is {\it
representable} (resp., {\it circularly representable}) if there
exists a non-empty word (resp., circular word) that witnesses
$S$.  Let $R_n$ denote the set of all non-empty representable subsets
of $\Sigma^n$, and let $\cip_n$ denote the set of all non-empty circularly
representable subsets of $\Sigma^n$.

Let $\sw(S)$ (resp., $\scw(S)$) denote the length of the
shortest non-circular witness (resp., circular witness) for $S$.
Let $\mu_n$ (resp., $\nu_n$) denote the maximum length of the
shortest non-circular (resp., circular) witness over all
representable subsets of $\Sigma^n$.

A {\it de Bruijn word} $B_n$ of order $n$ over the alphabet $\Sigma$ is
a shortest circular witness for the set $\Sigma^n$.  It is
known~\cite{debruijn1946}  that the length of a de Bruijn word of
order $n$ over $\Sigma$ is $2^n$.

For convenience, we let $w[i]$ denote the $i$'th letter of $w$ and
$w[i .. j]$ denote the factor of $w$ with length $j - i + 1$ that
starts with the $i$'th letter of $w$.  Thus $w = w[1..n]$ where
$n = |w|$.

\section{Bounds on the size of $\cip_n$}\label{sec:bounds}
In this section, we give lower and upper bounds on the size of $\cip_n$, both of which are of the form $\alpha^{2^n}$. Our lower bound has $\alpha = \sqrt{2}$ while our upper bound has $\alpha = \root4\of{10}$. Note that our lower bound also works for the size of $R_n$, since every circularly representable subset is also representable.

\subsection{Lower bound}

Our argument for the lower bound derives from constructing a set of circularly representable subsets.

\begin{proposition}
Let $b_n$ be any de Bruijn word of order $n$. Then
$| C_{n + 1}(b_n) | = 2^n$.
\end{proposition}

\begin{proof}
Every de Bruijn word of order $n$ is of length $2^n$; thus there are
$2^n$ length-$(n + 1)$ factors of $b_n$ (considered circularly).
These length-$(n + 1)$ factors are pairwise
distinct, for if $w \in \Sigma^{n + 1}$ appears more than once as a
factor of $b_n$, then $w[1..n]$ appears more than once as a factor of
$b_n$.  However, every length-$n$ factor appears
only once in $b_n$, a contradiction.  Hence $|C_{n + 1}(b_n)| = 2^n$.
\qed
\end{proof}

\begin{lemma}\label{lemma:element}
Given a de Bruijn word $b_n$, let $Y$ denote the set $\Sigma^{n + 1} \backslash C_{n + 1}(b_n)$. For any $y \in Y$, the set $\{y\} \cup C_{n + 1}(b_n)$ is circularly witnessed by a word $w$ for which both the length-$2^n$ prefix and the length-$2^n$ suffix equal $b_n$.
\end{lemma}

\begin{proof}
We construct such a witness for $\{y\} \cup C_{n + 1}(b_n)$.

Let $t = b_nb_nb_nb_n$. Let $y_1 = y[1..n]$ and $y_2 = y[2..n + 1]$.
Let $i_1$ denote the index of the first occurrence of $y_1$ in $t$;
namely, the index $i_1$ is the minimal integer such that $y_1 =
t[i_1..i_1 + n - 1]$. Let $i_2$ denote the index of the last occurrence
of $y_2$ in $t$; namely, the index $i_2$ is the maximal integer such
that $y_2 = t[i_2..i_2 + n - 1]$.

We argue that the first occurrence of $y_1$ does not overlap the last
occurrence of $y_2$. We have $i_1 \leq 2^n$, since every possible
factor of length $n$ appears in the circular word $b_n$. Similarly, we
obtain $i_2 > 3 \cdot 2^n - n$. Thus we have
$$i_1 + n - 1 - i_2 < -2 \cdot 2^n + 2n - 1 < 0,$$
and hence the first occurrence of $y_1$ does not
overlap the last occurrence of $y_2$.

Now consider the circular word
$$t_y = b_n b_n t[1..i_1 - 1] t[i_1..i_1 + n - 1] t[i_2 + n - 1]
t[i_2 +n..2^{n + 2}] b_n b_n.$$
We argue that $t_y$
is a witness for $\{y\} \cup C_{n + 1}(b_n)$. For one direction, every
element of $\{y\} \cup C_{n + 1}(b_n)$ appears as a length-$(n + 1)$
factor of $t_y$. This is a consequence of the following two facts:
\begin{enumerate}
\item{$b_nb_n$ witnesses $C_{n + 1}(b_n)$.}
\item{$t[i_1..i_1 + n - 1]t[i_2 + n - 1] = y[1..n]y[n + 1] = y$.}
\end{enumerate}
For the other direction, we can see that all factors of length $n + 1$ in $t_y$ are elements of $\{y\} \cup C_{n + 1}(b_n)$ by inspection.
Note that the length-$2^n$ prefix and the length-$2^n$ suffix of $t_y$ both equal $b_n$. Hence we conclude that there exists a word for which the prefix and the suffix equal $b_n$ and this circular word circularly witnesses $\{y\} \cup C_{n + 1}(b_n)$.
\qed
\end{proof}


\begin{example}
Let $n = 2$. One of the de Bruijn words of order $2$ is $b_2 = 0011$. We have $C_3(b_2) = \{001, 011, 110, 100\}$. Thus $Y = \{000, 010, 101, 111\}$. Let $y = 010$.
The following circular word demonstrates
that the set $\{y\} \cup C_{n + 1}(b_n)$ is representable:
$$t_{010} = (\underbrace{00110011}_{b_2b_2})(\underbrace{0}_{t[1..i_1 -
1]})(\underbrace{01}_{t[i_1 .. i_1 + n - 1] =
y_1})(\underbrace{0}_{t[i_2 + n - 1]})(\underbrace{011}_{t[i_2 +n..2^{n
+ 2}]})(\underbrace{00110011}_{b_2b_2}).$$
\end{example}

\begin{proposition}\label{pro:set}
Given a de Bruijn word $b_n$, let $Y$ denote the set $\Sigma^{n + 1} \backslash C_{n + 1}(b_n)$. For any subset $S \subseteq Y$, the set $S \cup C_{n + 1}(b_n)$ is a circularly representable subset of $\Sigma^{n + 1}$.
\end{proposition}

\begin{proof}
We have proved this proposition for the case where $|S| = 1$ by Lemma~\ref{lemma:element}. Now we turn to the general case. Let $S = \{s_1, s_2, \ldots, s_m\}$.
By Lemma~\ref{lemma:element}, for each $1 \leq i \leq m$, there exists a circular word $t_i$ that witnesses $\{s_i\} \cup C_{n + 1}(b_n)$ and both the prefix and the suffix of $t_i$ equal $b_n$. We argue that the circular word $t_S = t_1t_2\cdots t_m$ witnesses $S \cup C_{n + 1}(b_n)$.

First, for any $1 \leq i \leq m$, $s_i$ appears in $t_i$ and thus in $t_S$. Moreover, every element of $C_{n + 1}(b_n)$ appears in the prefix of $t_S$: $b_nb_n$. Thus, it suffices to show that every length-$(n + 1)$ factor of $t_S$ is a member of $S \cup C_{n + 1}(b_n)$. This is shown by the fact that for any $1 \leq i < m$, both the suffix of $t_i$ and the prefix of $t_{i + 1}$ equal $b_n$, which implies that the concatenation of $t_i$ and $t_{i + 1}$ does not produce any new factor of length $n + 1$ in $t_S$.

Thus, we conclude that for any subset $S$ of $Y$, there exists a witness for the set $S \cup C_{n + 1}(b_n)$.
\qed
\end{proof}

\begin{corollary}
A lower bound for the size of $\cip_{n + 1}$ is $2^{2^n} =
\sqrt{2}^{\, 2^{n+1}}$.
\end{corollary}

\subsection{Upper bound}

An obvious upper bound for $|\cip_n|$ is $2^{2^n}$, since $\cip_n \subseteq 2^{\Sigma^n}$, where $|2^{\Sigma^n}| = 2^{2^n}$. In this section, we will show that a tighter upper bound is $\alpha^{2^n}$, where $\alpha = \root4\of{10}$.\\

\begin{definition}
Let $S \subseteq \Sigma^{n + 1}$ and $T \subseteq \Sigma^n$. We say
that $S$ \emph{is incident on} $T$ if there exists a circular word
$w$ such that $w$ witnesses both $S$ and $T$.
\end{definition}

\begin{example}
For example, we fix $n = 4$. Let $w = 0110$. Then $w$ is a witness for the set $S = \{0110, 1100, 1001, 0011\} \in \cip_4$ and $T = \{011, 110, 100, 001\} \in \cip_3$. It follows that $S$ is incident on $T$.  Note that $w' = 01100110$ is also a witness for $S$, and a witness for $T$ as well.
\end{example}

In fact we can argue that if $S$ is incident on $T$, then every word that witnesses $S$ also witnesses $T$.

\begin{proposition}
Every set $S \in \cip_{n + 1}$ is incident on exactly one set in $\cip_n$.
\label{pro:incident}
\end{proposition}

\begin{proof}
Let $T = \{t \in \Sigma^n: \exists w \in S \text{ such that } t \text{ is a length-}n \text{ prefix or suffix of } w\}$. Then a word $w$ which witnesses $S$ also witnesses $T$. Thus $S$ is incident on $T$. Moreover, if $S$ is incident on $T$ and $T'$, then every witness of $S$ must also witness $T$ and $T'$. Thus we have $T = T'$. So we conclude that every set $S \in \cip_{n + 1}$ is incident on exactly one set in $\cip_n$.
\qed
\end{proof}

Now we give a partition of $\cip_{n + 1}$. Let
$$\cip_{n + 1}[T] = \{S
\in \cip_{n + 1}: S \text{ is incident on } T\}.$$
Proposition~\ref{pro:incident} implies that $\{\cip_{n + 1}[T]\}_{T \in
\Sigma^n}$ is a pairwise disjoint partition of the set $\cip_{n + 1}$.
Namely, (1) for every $T_1 \neq T_2$, we have $\cip_{n + 1}[T_1] \cap
\cip_{n + 1}[T_2] = \emptyset $ and (2) $\bigcup_{T \in \cip_n}\cip_{n
+ 1}[T] = \cip_{n + 1}$.\\

Thus we have $|\cip_{n + 1}| = \sum_{T \in \Sigma^n}|\cip_{n + 1}[T]|$. So to give an upper bound for $|\cip_{n + 1}|$, it suffices to give a upper bound for the size of $\cip_{n + 1}[T]$.

\begin{definition}
Let $x$ be a word of length $n$. We say that $P_x = \{ 0x, 1x\}$ is a
{\it pair} of order $n$ w.r.t $x$, that $S_x = \{0x, 1x, x0, x1\}$ is a {\it
skeleton} of order $n$ w.r.t. $x$, and $N_x = \{0x0, 0x1, 1x0, 1x1\}$
is a {\it net} of order $n$ w.r.t $x$. We also say that a set $S$
contains $P_x$ (resp., $S_x$ and $N_x$) if $P_x \subseteq
S$ (resp., $S_x \subseteq S$ and $N_x \subseteq S$).
\end{definition}

For any $T \subseteq \Sigma^n$, let $\sigma(T)$ denote the number of skeletons of order $n - 1$ in $T$ and let $\rho(T)$ denote the number of pairs of order $n - 1$ in $T$. We have the following proposition:
\begin{proposition}
For any $T \subseteq \Sigma^n$,
we have $|\cip_{n + 1}[T]| \leq 7^{\sigma(T)}.$
\label{pro: upper_main}
\end{proposition}

Before giving the proof for Proposition~\ref{pro: upper_main}, we
introduce another definition.

\begin{definition}
A set $R$ is {\it feasible} for a set $T \subseteq \Sigma^n$ if there exists $S \in \cip_{n + 1}[T]$ such that $R \subseteq S$.
\end{definition}

We observe that $\Sigma^{n + 1} = \bigcup_{x \in \Sigma^{n - 1}}N_x$ and thus any subset $S \in \Sigma^{n + 1}$ is a disjoint union of subsets of nets of order $n - 1$. Formally, for any subset $S \in \Sigma^{n + 1}$, we have $S = \bigcup_{x \in \Sigma^{n - 1}}R_x$, where $R_x \subseteq N_x$.

\begin{proof}[of Proposition~\ref{pro: upper_main}]
Let $F_x$ denote the set of feasible subsets (for $T$) of the net $N_x$. If $S \in R_{n + 1}[T]$, then $S$ is a disjoint union of feasible subsets (for $T$) of nets. Thus we have $|R_{n + 1}[T]| \leq \prod_{x \in \Sigma^n} |F_x|$. In order to prove this proposition, it now suffices to show that for any $x \in \Sigma^{n - 1}$, the following condition holds.
\begin{itemize}
\item{if $S_x \subseteq T$, then $|F_x| \leq 7$;}
\item{otherwise $|F_x| \leq 1$.}
\end{itemize}

For any $x \in \Sigma^{n - 1}$, we consider all the possible feasible subsets of $N_x$. Let $F$ denote any feasible subset of $N_x$.

\begin{itemize}
\item{
For the first case where $S_x \subseteq T$, we have the following properties:
\begin{enumerate}
\item{Either $0x0 \in F$ or $0x1 \in F$ since $0x \in T$;}
\item{Either $1x0 \in F$ or $1x1 \in F$ since $1x \in T$;}
\item{Either $0x0 \in F$ or $1x0 \in F$ since $x0 \in T$;}
\item{Either $0x1 \in F$ or $1x1 \in F$ since $x1 \in T$.}
\end{enumerate}

Hence we have at most $7$ possible feasible subsets of $N_x$ which are listed as follows: \{$0x0$, $1x1$\}, \{$0x0$, $0x1, $$1x1$\}, \{$0x0$, $1x0$, $1x1$\},
\{$0x0$, $0x1$, $1x0$, $1x1$\}, \{$0x0$, $0x1$, $1x0$\}, \{$0x1$, $1x0$\}, \{$0x1$, $1x0$, $1x1$\}. Thus $|F_x| \leq 7$.
%
%
%
}

\item For the second case where $S_x \not\subseteq T$, we argue that $|F_x| \leq 1$. Without loss of generality, suppose $0x \not\in T$. It follows that:
\begin{enumerate}
\item{$0x0$ and $0x1$ cannot occur in $F$ since $0x \not\in T$;}
\item{$1x0 \in F$ if and only if $x0 \in T$;}
\item{$1x1 \in F$ if and only if $x1 \in T$;}
\end{enumerate}
Hence, $F$ is fixed. It follows that $|F_x| \leq 1$.

\end{itemize}

By finishing the argument on the above two cases, we conclude that $|\cip_{n + 1}[T]| \leq 7^{\sigma(T)}$.
\qed
\end{proof}

Now, we are close to the core part. Instead of computing the number of skeletons, which is quite complex, we consider the number of pairs.

\begin{proposition}\label{pro:upperbound}
The size of the set $|\cip_{n + 1}|$ is bounded by $10^{2^{n - 1}}$.
\end{proposition}

\begin{proof}
Let $L_{k, i}$ denote the number of subsets $T \in \cip_n$, such that $|T| = k$ and $\rho(T) = i$. There are in total $2^{n - 1}$ pairs in $\Sigma^n$, and we first choose $i$'s pairs from them. Then, we choose the other $k - 2i$ elements which do not form any pair from the remaining $2^{n - 1} - i$ elements. Thus, we have $$L_{k, i} = {2^{n - 1}\choose i} {2^{n - 1} - i\choose k - 2i} 2^{k - 2i}.$$

Note that $k \geq 2i$ since a set of $k$ elements can contain at most $\frac{k}{2}$ pairs and the term $L_{k, i}$ vanishes when $k - 2i > 2^{n - 1} - i$. Thus we have $$ |\cip_{n + 1}| = \sum_{T \in \Sigma^n}|\cip_{n + 1}[T]| \leq \sum_{k = 0}^{2^n} \sum_{i = 0 }^{\frac{k}{2}} L_{k, i} 7^i.$$ The inequality holds since we count the number of pairs instead of the number of skeletons and the number of pairs is always greater than or equal to the number of skeletons.
Then we can see that
\begin{equation*}
\begin{aligned}
|\cip_{n + 1}| & \leq \sum_{k = 0}^{2^n} \sum_{i = 0}^{\frac{k}{2}} {2^{n - 1}\choose i} {2^{n - 1} - i\choose k - 2i} 2^{k - 2i} 7^i
            & \leq \sum_{i = 0}^{2^{n - 1}} {2^{n - 1}\choose i} 7^i \sum_{k = 2i}^{2^n} {2^{n - 1} - i\choose k - 2i} 2^{k - 2i}
\end{aligned}
\end{equation*}
by writing $L_{k, i}$ in closed form. Note that
\begin{equation*}
    \sum_{k = 2i}^{2^n} {2^{n - 1} - i\choose k - 2i} 2^{k - 2i} = \sum_{k = 0}^{2^n - 2i} {2^{n - 1} - i\choose k} 2^k
                                                                 = \sum_{k = 0}^{2^{n - 1} - i} {2^{n - 1} - i\choose k} 2^k
                                                                 = 3^{2^{n - 1} - i}.
\end{equation*}

So we have $$|\cip_{n + 1}| \leq  \sum_{i = 0}^{2^{n - 1}} {2^{n - 1}\choose i} 7^i 3^{2^{n - 1} - i}  = 10^{2^{n - 1}}. $$
\qed
\end{proof}

Proposition~\ref{pro:upperbound} directly implies the upper bound we claimed in the beginning of this section.

\section{Shortest witness}

Recall that $\mu_n$ (resp., $\nu_n$) is the maximum length of
the shortest non-circular witness (resp., circular witness) over
all subsets of $\Sigma^n$. The quantities of $\mu_n$ and $\nu_n$ are of
interest since we can enumerate all sequences of length less than or
equal to $\mu_n$ (resp., $\nu_n$) in order to list all the
representable (resp., circularly representable) subsets of
$\Sigma^n$.  In this section we obtain an upper bound on
$\mu_n$ and $\nu_n$.

We need the following result of Hamidoune \cite[Prop.\ 2.1]{hamidoune1978}.
Since the result is little-known and has apparently not appeared
in English, we give the proof here.  By a {\it Hamiltonian walk} we mean
a closed walk, possibly repeating vertices and edges, that visits every
vertex of $G$.

\begin{proposition}
Let $G = (V,E)$ be a directed graph on $n$ vertices.  If $G$ is strongly
connected (that is, if there is a directed path from every vertex to every
vertex), then there is a Hamiltonian walk of length
at most $\lfloor (n+1)^2/4 \rfloor$.  Furthermore, this bound is
best possible.
\end{proposition}

\begin{proof}
Let $L$ be a longest simple path in $G$. (A simple path does not repeat
edges or vertices.) Let $V-L = \{ v_i \, : \, 1 \leq i \leq  k \}$.
Let $v_0$ be the last
vertex in $L$ and $v_{k+1}$ be the first vertex in $L$. Let $L_i$
be a simple path
from $v_i$ to $v_{i+1}$. Then a Hamiltonian walk $W$
is obtained by following the
edges in $L_0, L_1, \ldots , L_k$, and then those in $L$.
So the number of edges
in $W$ is at most $(k+2)|L| = |L|(n+1-|L|)$. But it is easy to see that
$r(n+1-r)$ is maximized when $r = \lceil n/2 \rceil$, so
$r(n+1-r) = \lfloor (n+1)^2/4 \rfloor$, as claimed.

To see that this bound is best possible, consider a graph where there is a
directed chain of $\lfloor n/2 \rfloor$ vertices, where the last vertex has a
directed edge to $\lceil n/2 \rceil$ other vertices, and each of those vertices
have a single directed edge back to the start of the chain. The
shortest walk covering all the vertices traverses the chain, then an
edge to one of the other vertices, then a single edge back, and repeats
this $\lceil n/2 \rceil$ times.
The total length is then $(\lfloor n/2 \rfloor +1) \lceil n/2 \rceil
= \lfloor (n+1)^2/4 \rfloor$. So the bound is tight.
\qed
\end{proof}

From this we immediately get

\begin{proposition}
An upper bound for $\mu_n$ and $\nu_n$ is $2^{2n-2} + 2^{n-1}$.
\end{proposition}

\section{Numerical results}

It is not feasible to enumerate every single word to verify whether a subset is circularly representable (or non-circularly representable). For this reason, we exploit ideas from graph theory.

Formally, we define $G_n = (V_n, E_n)$, where
\begin{equation*}
\begin{aligned}
V_n & = \{(S, u, v): S \subseteq \Sigma^n \text{ and } u, v \in \Sigma^n\} \text{ and} \\
E_n & = \{((S, u, v), (S \cup \{x\}, u, x)): S \subseteq \Sigma^n, \text{ } u, v, x\in \Sigma^n, \text{ and }v[2..n] = x[1..n - 1]\}.
\end{aligned}
\end{equation*}
 We say that a node $(S, u, v)$ is {\it valid} if $S$ is witnessed by a non-circular word $w$ for which the length-$n$ prefix is $u$ and the length-$n$ suffix is $v$.

We use a breadth-first search strategy to compute all the possible valid nodes in $G_n$. Let $I$ denote a subset of nodes $\{(\{u\}, u, u): a \in \Sigma^n\}$ in $G_n$. Nodes in $G_n$ that are connected to any node in $I$ can be proven valid by induction. Thus, a breadth-first search begins with the subset $I$ and enumerates all nodes that are connected to nodes in $I$.

The relation between valid nodes in $G_n$ and non-empty representable subsets of order $n$ is that any subset $S \subseteq \Sigma^n$ is representable if and only if
there exist $u, v \in \Sigma^n$ such that $(S, u, v)$ is valid. This relation can be proved by induction. Similarly, any subset $S \subseteq \Sigma^n$ is circularly representable if and only if there exists $u \in \Sigma^n$ such that $(S, u, u)$ is valid and the minimum distance $d$ between $(S, u, u)$ and nodes in $I$ satisfies the inequality $d \geq n - 1$.

With the above properties, we can enumerate all the possible non-empty representable (or circularly representable) subsets of order $n$. Our results are shown in the following table.  The last two columns give words $w$
of length $\nu_n$ (resp., $\mu_n$) for which no shorter word
witnesses $C_n (w)$ (resp., $F_n (w)$).

\begin{center}
\begin{tabular}{|c|c|c|c|c|c|c|}
  \hline
  $n$ & $|\cip_n|$ & $|R_n|$ & $\nu_n$ & $\mu_n$ & longest circ. witness & longest witness\\
  \hline
  1 & 3 & 3 & 2 &  2 & 01 & 01  \\
  \hline
  2 & 6 & 14 & 4 &  5 & 0011 & 00110 \\
  \hline
  3 & 27 & 121 & 9 &  10 & 000100111 & 0001011100 \\
  \hline
  4 & 973 & 5921 & 24 &  24 & \footnotesize{000010001011100011101111 } & \footnotesize{000010010101100101101111} \\
  \hline
  5 & 2466131 & 20020315 & 82 & 77 & --- & --- \\
  \hline
\end{tabular}
\end{center}

\section{Fixed-length witnesses}

We now turn to a related question.  We fix a length $n$ and we ask, how
many different subsets of $\Sigma^n$ can we obtain by taking the
(ordinary, non-circular factors) of a word of length $t$?  We call this
quantity $T(t,n)$.  As we will
see, for $t < 2n$, there is a relatively simple answer to this
question.

In order to compute $T(t, n)$, we consider the number of words that
witness the same subset of $\Sigma^n$. Suppose $S \subseteq \Sigma^n$.
Let $C_t(S)$ denote the number of words of length $t$ that witness $S$.
Then we have
$$T(t, n) = 2^t - \sum\limits_{{S \in \Sigma^n} \atop  {C_t(S) >
1}} (C_t(S) - 1).$$
It suffices to characterize what subsets $S$
satisfy $C_t(S) > 1$ and to determine $C_t(S)$.

For $t < 2n$, we have such a characterization by Theorem~\ref{strong_pro}
below. Before stating the proposition, we first introduce some notation.

Let $w$ be a word. Let $\pref(w)$ denote the set of prefixes of $w$. A \textit{period} $p$ of $w$ is a positive integer such that $w$ can be factorized as
$$w = s^ks', \text{ with }|s| = p, \text{ }s' \in \pref(s), \text{ and } k \geq 1.$$
Let $\pi(w)$ denote the minimal period of $w$.

The \textit{root} of a word $w$ is the prefix of $w$ with length
$\pi(w)$. Let $r(w)$ denote the root of $w$. Two words $w$ and $w'$ are
{\it conjugate} if there exist $u, v \in \Sigma^{*}$ such that $w = uv$
and $w' = vu$; $w$ and $w'$ are \textit{root-conjugate} if their roots
$r(w)$ and $r(w')$ are conjugate.

The following theorem is crucial for our work and of independent interest.

\begin{theorem}
\label{strong_pro}
Let $t, n, k$ be such that $t = n + k$, $n \geq k + 1$, and $k \geq
0$.  Let $w$ and $w'$ be distinct words of length $t$ over an
arbitrary alphabet.  Then $F_n(w) = F_n(w')$ iff $\pi(w) = \pi(w') \leq
k + 1$ and $w, w'$ are root-conjugate.
\end{theorem}

One direction is easy:  if $w$ and $w'$ are root-conjugate with period
$p \leq k+1$, then there are $p$ places to begin,
and considering consecutive factors of length $n+p-1$ gives
exactly $p$ distinct length-$n$ factors.

For the other direction,
we need three lemmas.

\begin{lemma} \label{lemma2}
(Fine-Wilf theorem~\cite[Theorem 1]{fine1965uniqueness})~Let $w_1, w_2$ be two words. If $w_1$ and $w_2$ have a common prefix of length $\pi(w_1) + \pi(w_2) - 1$, then $r(w_1) = r(w_2)$.
\end{lemma}

\begin{lemma} \label{lemma3}
For any $w \in \Sigma^+$, if there exists a factorization $w = xyz$ such
that $xy = yz$ and $x, y, z \in \Sigma^+$, then $w$ is periodic with
$\pi(w) \leq |x|$.
\end{lemma}

\begin{proof}
By the Lyndon-Sch\"{u}tzenberger theorem~\cite[Lemma 2]{lyndon1962equation}, there exist $u \in \Sigma^+, v \in \Sigma^*$ and an integer $e \geq 0$ such that $x = uv, y = (uv)^eu, z = vu$. Thus $w = (uv)^{e + 2}u$. Thus $w$ is periodic with $\pi(w) \leq |x|$.
\qed
\end{proof}

\begin{lemma} \label{lemma4}
Let $t, n, k$ be integers such that $t = n + k$, $n \geq k + 1$, and $k
\geq 0$. Let $w$ be a word of length $t$ with $\pi(w) \leq k + 1$. If
$w'$ is any word such that $F_n(w) = F_n(w')$, then $w$ and $w'$ are
root-conjugate.
\end{lemma}

Carpi and de Luca proved a stronger proposition~\cite[Proposition
6.2]{carpi2003semiperiodic} which directly implies this lemma. We first
introduce some relevant notation from that paper.

A factor $s$ of a word $w$ is said to be \textit{right-special} in $w$
if there exist two distinct symbols $a$ and $b$ such that $sa$ and $sb$
are factors of $w$.  Let $R_w$ denote the minimal length $m$ such that
there exists no factor of length $m$ that is right-special.

A factor $s$ of a word $w$ is said to be \textit{right-extendable}
(resp., \textit{left-extendable}) in $w$ if there exists a symbol $a$
such that $sa$ is a factor of $w$ (resp., $as$ is a factor of $w$). Let
$K_w$ and $H_w$ denote the length of the shortest factor which is not
right-extendable (resp., left-extendable).

A word is \textit{semiperiodic} if $R_w < H_w$.

\begin{proof}[of Lemma~\ref{lemma4}]
Carpi proved \cite[Lemma 3.2]{carpi2003semiperiodic} that $\pi(w) > R_w$. Also, we have $H_w \geq \pi(w)$ since the length-($\pi(w) - 1$) prefix of $w$ is left-extendable. Thus $w$ is semiperiodic. Moreover we have $F_n(w) = F_n(w')$ where $n \geq k + 1 \geq \pi(w) \geq 1 + R_w$. Then we can apply \cite[Proposition 6.2]{carpi2003semiperiodic} to prove this lemma.
\qed
\end{proof}

\begin{proof}[of Theorem~\ref{strong_pro}]
We give a proof for Theorem~\ref{strong_pro} by induction on $k$.

The base case is when $k = 0$. In this case $t = n$ and thus $F_n(w) = \{w\}$ and $F_n(w') = \{w'\}$. Thus $w = w'$.

Now we deal with the induction step. We assume the result holds for $k - 1$ and we prove it for $k$. For convenience, we let $p_i(w)$ denote the length-$i$ prefix of the word $w$; let $s_i(w)$ denote the length-$i$ suffix of the word $w$.

We first consider the case where $H_w < n$. We have $p_n(w) \in F_n(w) = F_n(w')$. If $p_n(w) \neq p_n(w')$, then there exists $a \in \Sigma$ such that $ap_{n - 1}(w) \in F_n(w')$. Thus we have $ap_{n - 1}(w) \in F_n(w)$ which leads to the contradiction that $H_w \geq |ap_{n - 1}(w)| = t$. Hence $p_n(w) = p_n(w')$.

Now let $s = w[2 .. t]$ and $s' = w'[2 .. t]$. Clearly $|s| = |s'| = t
- 1$. The prefix $p_n(w)$ appears only once as a factor of $w$,
otherwise $p_{n - 1}(w)$ is left-extendable in $w$ which contradicts
the fact that $H_w < n$. Thus we have $F_n(s) = F_n(w) \backslash
\{p_n(w)\}$. Similarly we have $F_n(s') = F_n(w') \backslash
\{p_n(w)\}$. Thus $F_n(s) = F_n(s')$. Let $k' = k - 1$. We have $t - 1
= n + k - 1 = n + k'$ and $p \geq k + 1 > k' + 1$. By induction, we
have either
\begin{enumerate}
\setlength{\itemindent}{20pt}
\setlength{\itemsep}{1pt}
\item[Case 1:]{
 $s = s'$; or
}
\item[Case 2:]{
 $s$ and $s'$ are root-conjugate and $\pi(s) = \pi(s') = \rho$, where $\rho \leq k' + 1 = k$.
}
\end{enumerate}
In Case 1, it follows that $w = w'$, contradicting the fact that $w,
w'$ are distinct. In Case 2, we prove that $s = s'$ by showing that
their roots are identical. Suppose $s$ and $s'$ have a common prefix of
length $d$. We have $d \geq n - 1$, since $w$ and $w'$ have a common
prefix of length at least $n$. If $d \geq \rho$, then the root of $s$
is identical to the root of $s'$. Otherwise, we have the chain of
inequalities $k \geq \rho \geq d + 1 \geq n \geq k + 1$, which is
trivially a contradiction. Thus neither Case 1 nor Case 2 can occur and
we are done with the case where $H_w < n$.

Similarly we can prove the induction step when $K_w < n$. Thus it suffices to consider the case where $H_w \geq n$ and $K_w \geq n$. We first claim $\pi(w) \leq k + 1$. There are several cases to settle:
\begin{itemize}
\item The first case is when $p_{n - 1}(w) = s_{n - 1}(w)$ and the
occurrence of $p_{n - 1}(w)$ and $s_{n - 1}(w)$ do not overlap; namely
we have $w = p_{n - 1}(w)Lp_{n - 1}(w)$, where $L \in \Sigma^*$. We
have the inequality $n + k = t = |w| = 2|p_{n - 1}(w)| + |L| = 2(n - 1)
+ |L|$. Thus $|L| = k + 2 - n$. Hence $\pi(w) \leq |p_{n - 1}(w)L| = n
- 1 + k + 2 - n = k + 1$.

\medskip

\item The second case is when $p_{n - 1}(w) = s_{n - 1}(w)$ and these
occurrences overlap. Formally we put it as follows: there exist $x, y,
z \in \Sigma^+$, such that $p_{n - 1}(w) = xy = yz$ and $w = xyz$. It
follows that $\pi(w) \leq |x| \leq k + 1$ by Lemma~\ref{lemma3}.

\medskip

\item The last case is when $p_{n - 1}(w) \neq s_{n - 1}(w)$. Let $i_p$
denote the index of the last occurrence of $p_{n - 1}(w)$; namely $i_p
= \sup\{i \geq 0: p_{n - 1}(w) = w[i .. i + n - 2]\}$. Note that $i_p >
0$ since $p_{n - 1}(w)$ is left-extendable and $i_p \leq t - n + 2$
since $p_{n - 1}(w) \neq s_{n - 1}(w)$. Thus, the first occurrence of
$p_{n - 1}(w)$ (the prefix of $w$) overlaps the last occurrence of
$p_{n - 1}(w)$. By Lemma~\ref{lemma3}, we get that $w_1 = w[1 .. i_p +
n - 2]$ is periodic with $\pi(w_1) \leq i_p - 1$. Similarly we let
$i_q$ denote the index of the first occurrence of $s_{n - 1}(w)$ and
$w_2 = w[i_q..t]$. We have $0 < i_q \leq t - n + 2$ and $\pi(w_2) \leq
t - n + 2 - i_q$. The factors $w_1$ and $w_2$ overlap for at least
$|w_1| + |w_2| - t \geq \pi(w_1) + \pi(w_2) - 1$ symbols. Let $D$
denote the overlap of $w_1$ and $w_2$. We have $|D| \geq \pi(w_1) +
\pi(w_2) - 1$. Also $\pi(w_1)$ is a period of $D$ since $|D| \geq
\pi(w_1)$ and $D$ can be factorized as
$$D = d^ld', \text{ where } d
\text{ is conjugate to the root of }w_1, \text{ }d' \in \pref(d),
\text{ and } l \geq 1.$$
By Lemma~\ref{lemma2}, the overlap $D$ has the
same root as $w_2$. Since root-conjugacy is an equivalence relation, we
have $w_1$ and $w_2$ are root-conjugate. It follows that $w$
is periodic with $\pi(w) = \pi(w_1) \leq k + 1$.

\end{itemize}

Finally by Lemma~\ref{lemma4}, we get that $w$ and $w'$ are
root-conjugate and their periods $\pi(w) = \pi(w') \leq k + 1$.  By all
cases, we finish the induction and complete the proof of
Theorem~\ref{strong_pro}.
\qed
\end{proof}

The following corollary gives $T(t, n)$ when $t < 2n$.

\begin{corollary}
For $n \leq t < 2n$, we have $T(t, n) = 2^t - \sum\limits_{k = 1}^{t - n + 1}\frac{k - 1}{k}\sum\limits_{d|k}\mu(\frac{k}{d})2^d$, where $\mu(\cdot)$ is the M\"{o}bius function.
\label{tcor}
\end{corollary}

\begin{proof}
Let $k = t - n$. We have $n \geq t - n + 1 = k + 1$. By
Theorem~\ref{strong_pro}, we know that for any set $S \subseteq
\Sigma^n$, $C_t(S) > 1$ if and only if there exists a word $w$ that
witnesses $S$ with $\pi(w) \leq k + 1$. In this case we have $C_t(S) =
\pi(w)$; that is, the set of words that witness $S$ is the same as the
set of the words that are root-conjugate to $w$. Thus each $S$ such
that $C_t(S) > 1$ corresponds to a set of root-conjugate words, which
can be represented by their lexicographically least roots (the Lyndon
words).

Thus we have
\begin{equation*}
\begin{aligned}
T(t, n) & = 2^t - \sum\limits_{{S \in \Sigma^n} \atop { C_t(S) > 1}} (C_t(S) - 1)
         = 2^t - \sum\limits_{ {\text{$w$ is a Lyndon word} \atop
{ \pi(w) \leq k + 1}}} (\pi(w) - 1)\\
        & = 2^t - \sum_{i = 1}^{k + 1} (i - 1)\cdot L(i),
\end{aligned}
\end{equation*}
where $k = t-n$ and $L(i) = \frac{1}{i}\sum\limits_{d|i}\mu(\frac{i}{d})2^d$ is the number of Lyndon words of length $i$.
\qed
\end{proof}

\begin{example}
To finish this section, we give a table listing some numerical results
for $T(t, n)$.  The numbers in bold follow from Corollary~\ref{tcor}.

\vspace{2mm}
\begin{center}
\begin{tabular}{|c|r|r|r|r|r|r|r|r|r|r|r|r|r|r|r|r|}
  \hline
  \backslashbox{$n$}{$t$} & $1$ & $2$ & $3$ & $4$ & $5$ & $6$ & $7$ & $8$ & $9$ & $10$ & $11$ & $12$ & $13$ & $14$ & $15$ & $16$\\
  \hline
  $1$ &{\bf 2} & 3 & 3 & 3 &3 & 3 & 3 & 3 & 3 & 3 & 3 & 3 & 3  & 3 & 3 & 3 \\
  \hline
  $2$ &&{\bf 4}&{\bf 7}&11&12&12&12&12&12&12&12&12&12 & 12 & 12 & 12 \\
  \hline
  $3$ && &{\bf 8}&{\bf 15}&{\bf 27}&48&72&94&100&103&101&103&101 & 103 & 101 & 103  \\
  \hline
  $4$ && & &{\bf 16}&{\bf 31}&{\bf 59}&{\bf 114}&216&391&677&1087&1621&2246 & 2928 & 3595 & 4235  \\
  \hline
  $5$ && & & &{\bf 32}&{\bf 63}&{\bf 123}&{\bf 242}&{\bf 474}&933&1795&3421&6399  & 11682 & 20704 & 35914  \\
  \hline
  $6$ && & & & &{\bf 64}&{\bf 127}&{\bf 251}&{\bf 498}&{\bf 986}&{\bf 1965} &3899&7709 & 15171 & 29710 & 57726 \\
  \hline
  $7$ && & & & & & {\bf 128} & {\bf 255} & {\bf 507} & {\bf 1010} & {\bf 2010} & {\bf 4013} & {\bf 8001} & 15969 & 31789 & 63256 \\
  \hline
  $8$ &&&&&&&& {\bf 256} & {\bf 511} & {\bf 1019} & {\bf 2034} & {\bf 4058} &
  {\bf 8109} & {\bf 16193} & {\bf 32367} & 64671 \\
  \hline
\end{tabular}
\end{center}
\end{example}

\section{Open Problems and Future Work}

\begin{enumerate}

\item In Section~\ref{sec:bounds}, we gave lower and upper bounds on
$|\cip_n|$, both of the form $\alpha^{2^n}$. Does the limit
$\lim\limits_{n \rightarrow \infty} |\cip_n|^{\frac{1}{2^n}}$ exist?

\item Find better bounds for $\mu_n$ and $\nu_n$.  For example, is
$\mu_n \leq (n-1)2^n$ for $n \geq 2$?

\item It is easy to see that
Theorem~\ref{strong_pro} fails for $t < k+1$.  Indeed, it is possible
to have $F_n (x) = F_n (y)$ in this case, and yet $\pi(x) \not= \pi(y)$.
For example, take $n = k-1$ so that $t = 2k-1$, and consider
$x = 0^k 1 0^{k-2}$ and $y = 0^{k-1} 1 0^{k-1}$.
Then $F_n (x) = F_n (y)$ but $\pi(x) = k+1$ and $\pi(y) = k$.

The remaining case is $n = k$, so that $t = 2k$.
We conjecture that if $x$ and $y$ are distinct binary words
of length $2n$ with $F_n (x) = F_n(y)$ then $\pi(x) = \pi(y)$ and
furthermore $x$ and $y$ are root-conjugate.  However, it is possible
in this case that $\pi(x) > n+1$.
Furthermore it seems that if $\pi(x) > n+1$, then
$x = u v 0 1 v^R u$ and $y = u v 1 0 v^R u$ (or vice versa)
for some nonempty words $u, v$ where $u$ is a palindrome and
$\pi(x) = n + |u|$.

As an example,
consider $x = 010110$, $y = 011010$.  Then
$F_3 (x) = F_3 (y) = \lbrace 010, 011, 101, 110 \rbrace$
but $\pi(x) = \pi(y) = 5$.    Here $u = 0$, $v = 1$.

\end{enumerate}

\end{document}